\documentclass{amsart}

\usepackage[a4paper,headheight=0.3cm]{geometry}
\usepackage{fancyhdr}
\usepackage{url}
\usepackage{tikz}
\usepackage{amsthm}
\usepackage{mathtools}
\usepackage{comment}
\usepackage{enumitem}
\usepackage{url}
\usepackage[pdftitle={Robust detection of exotic infectious diseases in animal herds: A comparative study of three decision methodologies under severe uncertainty},%
pdfauthor={Matthias C. M. Troffaes, John Paul Gosling},%
pdfkeywords={detection, decision, test, infectious, disease, info-gap, maximality, Bayesian, imprecise, maximin, minimax}]{hyperref}

\newtheorem{theorem}{Theorem}

\newcommand{\compl}[1]{{#1}^c}
\newcommand{\lpr}{\underline{P}}
\newcommand{\upr}{\overline{P}}
\newcommand{\reals}{\mathbb{R}}

\title[Robust detection of exotic infectious diseases in animal herds]{Robust detection of exotic infectious diseases in animal herds: A comparative study of three decision methodologies under severe uncertainty}
\author{Matthias C. M. Troffaes}
\address{Durham University, Dept. of Mathematical Sciences, Science Laboratories, South Road, Durham
DH1 3LE, United Kingdom} \email{matthias.troffaes@gmail.com}
\author{John Paul Gosling}
\address{University of Leeds, School of Mathematics,
Leeds, LS2 9JT, United Kingdom} \email{j.p.gosling@leeds.ac.uk}
\keywords{Exotic disease, lower prevision, info-gap, maximality, minimax, robustness, inspection, protocol}

\begin{document}

\begin{abstract}
  When animals are transported and pass through customs, some of them
  may have dangerous infectious diseases. Typically, due to the cost
  of testing, not all animals are tested: a reasonable selection must
  be made. How to test effectively whilst avoiding costly disease outbreaks?
  First, we extend a model proposed in the literature
  for the detection of invasive species to suit our purpose,
  and we discuss the main sources of model uncertainty, many of which are
  hard to quantify.
  Secondly, we explore and compare three decision methodologies on the problem at hand, namely,
  Bayesian statistics,
  info-gap theory and imprecise probability theory, all of which are
  designed to handle severe uncertainty.
  We show that, under rather general conditions, every
  info-gap solution is maximal with respect to a suitably chosen imprecise probability model, and that therefore, perhaps surprisingly, the set of maximal options can be
  inferred at least partly---and sometimes entirely---from an info-gap analysis.
\end{abstract}

\maketitle

\thispagestyle{fancy}

\section{Introduction}

This paper concerns the inspection of imported herds
of animals for signs of known or unknown major exotic infectious
diseases. Imports and exports of animals represent a significant
contribution to UK agriculture. Even though imports are subject to
strict controls at the UK border under EU and national rules, there is a real risk of animal diseases being introduced. F\`evre et al. \cite{2006:fevre} review the problems associated with animal movement and the spread of disease.

We will build further on the work of Moffitt
et al. \cite{2008:moffitt}, who study inspection protocols
for shipping containers of invasive species, employing info-gap
theory \cite{2001:benhaim} to model the severely uncertain number of
infested items.
The aim of their study is to realistically take into account economical
considerations (actual costs of testing, and of
invasive species passing through customs), whilst also soundly handling
the enormous uncertainty.

A key feature of their, and also our, problem is that exact probabilities of the constituent events are very hard
to come by \cite{2006:moffit:osteen}. This motivates the use of robust
uncertainty models and decision tools,
such as info-gaps \cite{2001:benhaim} (i.e. robust satisficing) as in
the original study,
but also
Bayesian statistics \cite{1985:berger}
and
imprecise probabilities \cite{2007:troffaes:decision:intro},
as we will do in this paper.

Our study, using each of these decision methodologies, leads us
to surmise a connection between info-gap analysis and imprecise
probability theory ($\Gamma$-minimax and maximality in particular). We
prove that the perceived connection is no coincidence, and we
establish a rigorous theoretical link between the two approaches.

The paper is organised as follows. Section~\ref{sec:probdesc}
introduces the problem of animal inspection, defines the model,
discusses various uncertainties involved, and derives an expression
for the expected loss under a simple binomial model for
infection. Section~\ref{sec:decision} solves the inspection problem,
first by Bayesian analysis, then using an info-gap model, and finally using an imprecise probability (or, robust Bayesian)
model with maximality.
These results are discussed in
Section~\ref{sec:discussion}, where we formally define an info-gap
model based on a nested set of imprecise probability models,
and establish the theoretical connections between 
info-gap,
$\Gamma$-minimax, and maximality.
Section~\ref{sec:conclusion}
concludes the paper.

\section{Animal Herd Testing}
\label{sec:probdesc}

In this section, we extend a model, proposed by \cite{2008:moffitt}  for
the detection of invasive species, to suit our purpose:
\begin{itemize}
\item we
explicitly take specificity and sensitivity into account in order to
allow for imperfect testing, 
\item we take into account an additional
cost term for terminating the herd in case an infection is detected,
and
\item we model the occurrence of diseased animals in the herd as a
binomial process, under a worst-case assumption of independence of infections between
animals.
\end{itemize}

\subsection{Model Description}

Consider a herd of $n$ animals, of which $m$ are tested---the problem
is to choose $m$ optimally.
The uncertain
number of diseased animals in the herd is denoted by $d$.
The test has sensitivity---the probability that a diseased animal
tests positive---equal to $p$, and specificity---the probability that
a healthy animal tests negative---equal to $q$.

Testing $m$ animals costs $c(m)$ utiles. If $d$ diseased animals
pass inspection undetected, we incur a cost of $a(d)$ utiles.
When at least one diseased animal is detected,
then, typically, the whole herd is terminated, costing $t(n)$ utiles.

Following \cite[p.~295, Sec.~3]{2008:moffitt},
in the numerical examples that follow, we take
\begin{align}
  c(m)&=1000-2000m+1000m^2 \\
  a(d)&=
  \begin{cases}
  0 & \text{if }d=0 \\
  a & \text{if }d\ge 1
  \end{cases}
  && (a=10\,000\,000)
\end{align}
Moffitt et al. \cite{2008:moffitt} consider $n$ between $250$ and $2\,500$,
do not need to consider the cost of termination ($t(n)=0$),
and assume perfect testing ($p=q=1$).
For our problem, in numerical examples that follow, we take
\begin{align}
  n&=250 \\
  t(n)&=400n=100\,000 \\
  p&=0.9999 \\
  q&=0.999
\end{align}
so we assume that a diseased animal
tests positive with probability $0.9999$,
and a healthy animal tests negative with probability $0.999$.
For
reference, if $q=0.999$, then probability that all animals in a healthy herd
of size $n=250$ test negative is $q^n=0.78$.

\subsection{Model Uncertainties}

Obviously, many of these values are rather uncertain.
The only values we are pretty certain of are
the number of animals $n$ in the herd,
the cost of testing $c(n)$,
and the cost of termination $t(n)$.

Due to the necessity that the herd must have valid health
documentation, we would expect that the number of infected animals $d$
would be low. Additional inspection by veterinary officials is costly
and depends on the inspecting official's ability to spot signs of
infectious disease like pathological lesions and abnormal
behaviour. Of course, the level of experience and competency will vary
from official to official, but the testing procedure should be
thorough enough for us to be confident of both a high sensitivity,
$p$, and specificity, $q$. In addition to this, the government would
prefer the most sensitive test possible (within budgetary
constraints), even if specificity was slightly compromised, because a
rare false positive would be better for the prevention of disease
entry than a rare false negative. Hence, we would expect
$p>q$. Further discussion of this can be found in \cite{1997:zeman}.

Of course, in general, having values for $p$ and $q$ as high as $0.9999$ and $0.999$
is unrealistic. For
most tests, the developers have aimed at getting a high value for the
specificity and sensitivity suffers. However, there are examples in
animal disease testing where both the sensitivity and specificity are
this high. For example, the virus antibody test for caprine
arthritis-encephalitis claims sensitivity and specificity values of
over 99.5\% \cite{2003:herrmann} and near perfect sensitivity and
specificity have been estimated for the polymerase chain reaction test
for parasites in fish \cite{2000:enoe}.

Regarding the cost $a$ of an infection passing through customs,
some historical data is available.
For example, instances of major disease outbreaks in the last couple of decades
include BSE where public spending was over \pounds 5 billion, and the
foot and mouth outbreak in 2001 which costed the UK government \pounds
2.6 billion \cite{2009:defra:animal:health}. These experiences show
that there is great variation in the level of costs of exotic disease
outbreaks. Due to the exceptional nature of the outbreaks, there is
limited evidence on which to base cost assessments. Therefore, there
is great uncertainty about what may happen in the future.

Outbreaks of any particular exotic disease are generally rare or may
never have occurred at all. Also, diseases change as new strains develop;
consequently, the possibility of new diseases arriving can change
rapidly. For example, until a few years ago, bluetongue was considered
extremely unlikely in some European countries, but now outbreaks are
expected every couple of years.

In late 2009, an elicitation exercise was carried out with government
experts to help quantify the average annual costs to the UK government
of exotic infectious disease outbreaks and the uncertainty about those
estimates \cite{2011:gosling:exotic:diseases}. In that exercise, it
was clear that the costs are severely uncertain even when the
disease was known (for example, foot and mouth is an exotic infectious
disease). A major contributor to the uncertainty about the overall
cost was the possibility of an outbreak of an unknown infectious
disease, which
could cost anywhere from \pounds 0.5 billion to \pounds 6 billion.

The scale and costs of an outbreak will depend on the length of time
between the diseased animal entering circulation and the disease's
presence being confirmed, and the speed and effectiveness of the
government's response. The eventual costs are influenced by any public
health implications and the effects of disease controls on other
industries. The main elements of the costs due to control measures
include: the disposal of and payments for culled animals; the tracing,
testing and diagnosis of animals; the cleaning and disinfection of
infected premises; and administrative costs in managing the
outbreak. The size of these costs will vary according to the scale of
the outbreak with key factors being the number of infected premises,
the numbers of animals culled, and the duration of the outbreak. These
types of factors are considered in greater detail in
\cite{2009:defra:animal:health} and \cite{1995:garner}.

A serious study of how all uncertainties involved could be taken into
account in the model would of course be extremely interesting, but is
beyond the goal of this paper. Instead, in this initial study, following
\cite{2008:moffitt} and many others, for now we will focus on the main
uncertainty, that is, the number of diseased animals $d$, and simply
assume reasonable values for the remaining parameters.
Such simplistic model helps to illustrate the
methodological differences and to motivate the theory.

\subsection{Expected Loss}

First, we derive the expected loss, in case all parameters of the
problem are perfectly known, including the number of diseased animals
$d$. Clearly, conditional on $d$, the expected loss is:
\begin{equation}
  L(m,d,p,q,c,a,t)=c(m)+t(n)\Pr(T|d)+a(d)\Pr(T^c|d)
\end{equation}
where $T$ denotes termination of
the herd, that is, the event that at least one diseased animal is
detected, and $T^c$ denotes its complement, that is, the event that
the herd passes inspection.

Let us deduce $\Pr(T^c|d)$.
First, if the test group of size $m$ is sampled randomly and without replacement,
then the probability of
exactly $z$ diseased animals in the test group follows a
hypergeometric distribution:
\begin{equation}
  \Pr(z|d)=\frac{\binom{d}{z}\binom{n-d}{m-z}}{\binom{n}{m}}.
\end{equation}

Next, we calculate the probability of non-termination
given $z$ diseased animals in the test group, that is
$\Pr(\compl{T}|d,z)$. If $d=0$, then the probability of
non-termination is the probability of all healthy animals in the
sample testing negative, so $\Pr(\compl{T}|0,z)=q^m$. If $d\ge 1$, then
given $z$ diseased animals in the sample, non-termination occurs when
none of the $z$ diseased animals tests positive and all of the $m-z$
healthy animals test negative. Hence, in all cases,
\begin{equation}
\Pr(\compl{T}|d,z) = (1-p)^z q^{m-z}.
\end{equation}
By the law of total probability,
\begin{align}
\Pr(\compl{T}|d)
\nonumber
&= \sum_{z=0}^{d} \Pr(\compl{T}|d,z)\Pr(z|d) \\
\label{eq:prob:nontermination}
&= \sum_{z=0}^{d} (1-p)^z q^{m-z}  \dfrac{\binom{d}{z}\binom{n-d}{m-z}}{\binom{n}{m}}.
\end{align}

Now we have all the ingredients to calculate the total expected loss
if we choose to test $m$ out of $n$ animals:
\begin{align}
  L(m,d,p,q,c,a,t)
  &=
  c(m)+t(n)+(a(d)-t(n))\Pr(\compl{T}|d)
  \\
  \intertext{or, if $a'(n,d)=a(d)-t(n)$ denotes the termination adjusted cost of apocalypse,}
  &=
  c(m)+t(n)+a'(n,d)\Pr(\compl{T}|d)
\end{align}
where $\Pr(\compl{T}|d)$ is given by Eq.~\eqref{eq:prob:nontermination}.
Figure~\ref{fig:loss} depicts the expected loss for a few typical
cases.

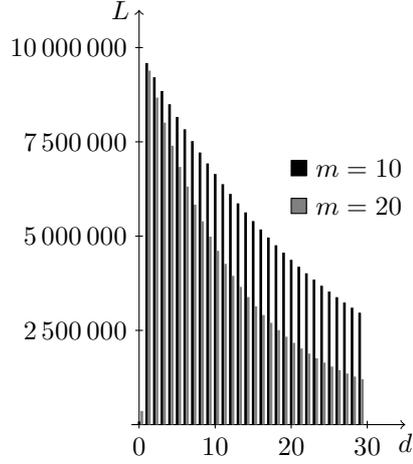
\begin{figure}
  \centering
  \begin{tikzpicture}[xscale=0.1,yscale=0.5]
    \draw[->] (-1,0) -- (35,0) node[below] {$d$};
    \draw[->] (0,-0.1) node[below] {$0$} -- (0,11) node[left] {$L$};
    \draw (10,0.1) -- (10,-0.1) node[below] {$10$};
    \draw (20,0.1) -- (20,-0.1) node[below] {$20$};
    \draw (30,0.1) -- (30,-0.1) node[below] {$30$};
    \draw (0.5,2.5) -- (-0.5,2.5) node[left] {$2\,500\,000$};
    \draw (0.5,5) -- (-0.5,5) node[left] {$5\,000\,000$};
    \draw (0.5,7.5) -- (-0.5,7.5) node[left] {$7\,500\,000$};
    \draw (0.5,10) -- (-0.5,10) node[left] {$10\,000\,000$};
    \draw[line width=1pt,color=black] plot[ycomb] file {infectious-loss-m10.table};
    \draw[line width=1pt,xshift=10pt,color=gray] plot[ycomb] file {infectious-loss-m20.table};
    \draw[fill=black] (20,7) -- (22,7) -- (22,6.6) node[midway,right] {$m=10$} -- (20,6.6) -- cycle;
    \draw[fill=gray] (20,6) -- (22,6) -- (22,5.6) node[midway,right] {$m=20$} -- (20,5.6) -- cycle;
  \end{tikzpicture}
  \caption{Loss as a function of the number of diseased animals for $m=10$ and $m=20$.}\label{fig:loss}
\end{figure}

\subsection{A Binomial Model for Infection}

Moffitt et al. \cite{2008:moffitt} consider an info-gap model directly
over the number of diseased animals $d$, which leads to a rather
tricky optimisation problem. Instead, we will consider the (highly
uncertain) probability $r$ that an animal is infected, and derive the
expected loss as a function of $r$. Although we do not explore this
topic further in this paper, this also paves the
way to modelling spatial dependencies between infections in the herd,
leading to more optimal testing strategies.

So, assume that each animal has a probability $r$ of being infected;
for simplicity, for now, we assume that one animal being diseased does
not affect another animal being diseased. Obviously, this will
generally not be satisfied, and more realistically, we would expect a
positive correlation, resulting in diseased animals being clustered
together in the herd. Assuming independence essentially
amounts to a worst case study: at the other extreme end, if one
diseased animal would immediately infect the whole herd, then it would
be sufficient to test only a single animal, as $d=0$ and $d=n$ would
be the only two possibilities.

Under the worst case assumption of independence, the probability of having $d$
out of $n$ animals infected is:
\begin{equation}\label{eq:prob:d:given:r:binom}
  \Pr(d|r)=\binom{n}{d} r^d(1-r)^{n-d}
\end{equation}
The expected loss is:
\begin{equation}
  E(L(m,\cdot,p,q,c,a,t)|r)=\sum_{d=0}^nL(m,d,p,q,c,a,t)\Pr(d|r)
\end{equation}
From now onwards, we will simply write $L(m|r)$ instead of $E(L(m,\cdot,p,q,c,a,t)|r)$ in order to simplify notation.
Figure~\ref{fig:expected:loss} depicts $L(m|r)$ as a function of $m$ for
a few typical situations.

\begin{figure}
  \centering
  \begin{tikzpicture}[xscale=0.08,yscale=3.47826086957]
\draw[->] (-0.625,0.2) -- (31.25,0.2) node[right] {$m$};
\draw[->] (0,0.185625) -- (0,2.64375) node[above] {$E(L|r) / 10^6$};
\node at (0,0.2) [below] {$0$};
\draw (25,0.185625) -- (25,0.214375) node[below] {$25$};
\node at (0,0.2) [left] {$0.2$};
\draw (-0.625,2.5) -- (0.625,2.5) node[left] {$2.5$};
\draw[solid, thick]  plot[smooth] file {infectious-expected-loss-0.table};
\draw[solid, thick]  plot[smooth] file {infectious-expected-loss-1.table};
\draw[solid, thick]  plot[smooth] file {infectious-expected-loss-2.table};
\draw[solid, thick]  plot[smooth] file {infectious-expected-loss-3.table};
\end{tikzpicture}
  \caption{Expected loss $L(m|r)$ as a function of the test group size
    $m$, for $r=0.00010$, $r=0.00025$, $r=0.00050$, and $r=0.00100$,
    from bottom to top.
    }
  \label{fig:expected:loss}
\end{figure}
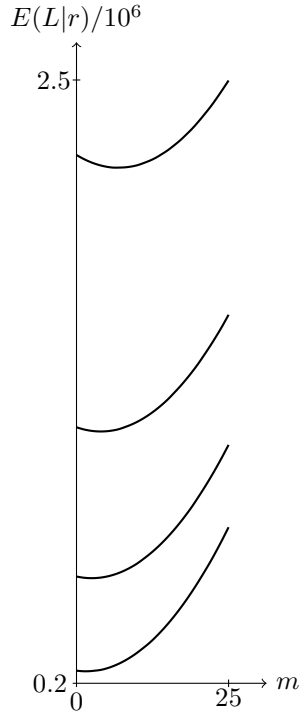

\section{Decision Analysis}
\label{sec:decision}

In this section, we explore and compare three decision methodologies,
  designed for severe uncertainty, on the problem at hand.
  In particular, 
\begin{itemize}
\item we perform a Bayesian analysis,
\item we accommodate the info-gap approach suggested by
  \cite{2008:moffitt} to our extended model,
\item we investigate possible ways of constructing sets of
  probabilities (i.e. imprecise probability models) which are in some
  sense equivalent to the proposed info-gap model, and
\item we compare the decisions that these various models lead to.
\end{itemize}

\subsection{Bayesian Analysis: Minimising Expected Loss}
\label{sec:bayesian}

In this approach, we model the uncertainty about $d$ probabilistically and choose the $m$ that minimises expected loss. Of course, this is not the only strategy open to us when we use a Bayesian formulation, but it is the most common.

Using the same model for $d$ as in Eq.~\eqref{eq:prob:d:given:r:binom}, we can write the probability mass function for $d$ to be
\begin{equation}
\Pr(d) = \int_0^1 \Pr(d|r) p(r) dr,
\end{equation}
where $p(r)$ is a probability density function chosen to characterise our beliefs about $r$. For computational convenience and because the density is sufficiently flexible, we choose to model $r\sim\mbox{Beta}(\alpha,\beta)$, where $\alpha$ and $\beta$ are chosen so that the Beta density matches our beliefs. Using this density for $r$,
\begin{equation}
\Pr(d|\alpha,\beta) = {{n}\choose{d}}\frac{\mbox{Be}(\alpha+d,\beta+n-d)}
{\mbox{Be}(\alpha,\beta)},
\end{equation}
where $\mbox{Be}(\cdot,\cdot)$ denotes the usual Beta function. In this analysis, we consider various choices of parameters for the Beta distribution, as listed in Table~\ref{tab:bayes:params}.
Following Walley \cite{1996:walley::idm}, $t$ denotes the expectation $\frac{\alpha}{\alpha+\beta}$ of the prior, and $s$ denotes $\alpha+\beta$, so $\alpha=st$ and $\beta=s(1-t)$.
\begin{table}
  \centering
  \begin{tabular}{rrrrr|rr}
$t$ & $s$ & $\sigma$ & $\alpha$ & $\beta$ & $m^*$ & $E(L|\alpha,\beta)/10^6$ \\
\hline
0.0002 & 199.0 & 0.001 & 0.040 & 198.9 & 2 & 0.316 \\
0.0004 & 398.8 & 0.001 & 0.160 & 398.7 & 3 & 0.738 \\
0.0008 & 798.4 & 0.001 & 0.639 & 797.7 & 6 & 1.567 \\
0.0016 & 1596.4 & 0.001 & 2.554 & 1593.9 & 10 & 3.002 \\
\end{tabular}

  \caption{Choices for parameters of the Beta distribution over $r$. For information, we also state the standard deviation $\sigma$, and the $\alpha$ and $\beta$ parameters of the canonical parametrization.}
  \label{tab:bayes:params}
\end{table}
In our analyses, we have varied the expectation of $r$ by varying $t$, and we have chosen $s$ such that the standard deviation is 0.001 throughout. We have chosen to investigate this set of distributions because they cover a range of reasonable beliefs that we may have about $r$ when we are dealing with rare diseases.


For each Beta distribution, we can calculate the expected loss for each choice of $m$ using
\begin{equation}
E(L|\alpha,\beta)=\sum_{d=0}^{n} \Pr(d|\alpha,\beta)L(m,d,p,q,c,a,t).
\end{equation}
Figure~\ref{fig:bayes:loss} is a plot of the expected loss against number of tests for the three chosen distributions. The minimum expected loss occurs for $m=m^*$, given in Table~\ref{tab:bayes:params}.

\begin{figure}
\centering
\begin{tikzpicture}[xscale=0.16,yscale=1.21212121212]
\draw[->] (-0.3125,0) -- (28.125,0) node[right] {$m$};
\draw[->] (0,-0.04125) -- (0,3.7125) node[above] {$E(L|\alpha,\beta) / 10^6$};
\node at (0,0) [below] {$0$};
\draw (25,-0.04125) -- (25,0.04125) node[below] {$25$};
\node at (0,0) [left] {$0$};
\draw (-0.3125,3.3) -- (0.3125,3.3) node[left] {$3.3$};
\draw[solid, thick]  plot[smooth] file {infectious-bayes-loss-m-0.table};
\draw[dashed, thick]  plot[smooth] file {infectious-bayes-loss-m-1.table};
\draw[dotted, thick]  plot[smooth] file {infectious-bayes-loss-m-2.table};
\draw[solid, thick]  plot[smooth] file {infectious-bayes-loss-m-3.table};
\end{tikzpicture}
\caption{Expected loss plotted against $m$ for each choice of parameters. Lower curves correspond to lower values of $t$.}
\label{fig:bayes:loss}
\end{figure}
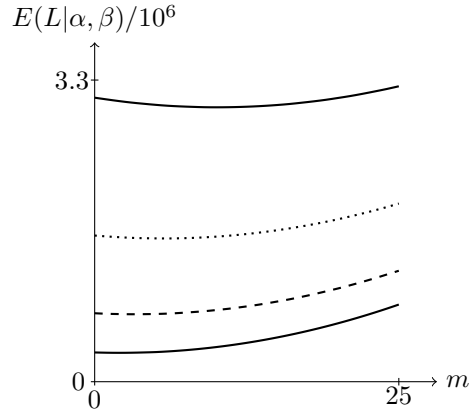

Also, by taking a probabilistic viewpoint, we can derive a distribution for the possible losses given a choice of $m$ and $p(r)$. If we choose $m=10$ and the worst case distribution  for $r$ (bottom row of Table~\ref{tab:bayes:params} where there is a greater than 30\% chance of there being diseased animals), we can find probabilities for the losses exceeding different values. For the thresholds in Figure~\ref{fig:bayes:solution}, as soon as we exceed the cost of termination and testing, we are essentially calculating the probability of disease outbreak. The jumps in this plot correspond precisely to the cost of testing ($81\,000$), and the cost of testing plus termination ($181\,000$).

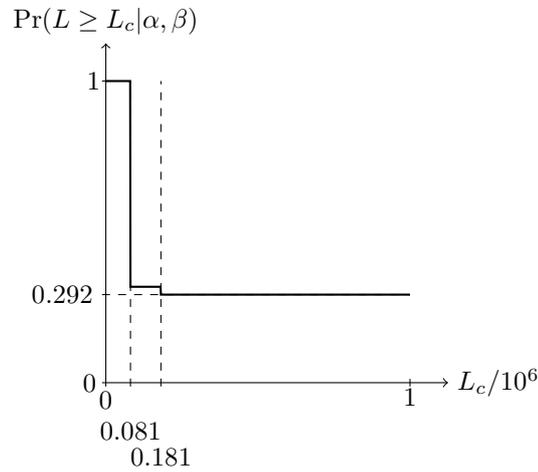
\begin{figure}
  \centering
  \begin{tikzpicture}[xscale=4.0,yscale=4.0]
\draw[->] (-0.0125,0) -- (1.125,0) node[right] {$L_c / 10^6$};
\draw[->] (0,-0.0125) -- (0,1.125) node[above] {$\Pr(L\ge L_c|\alpha,\beta)$};
\node at (0,0) [below] {$0$};
\draw (1,-0.0125) -- (1,0.0125) node[below] {$1$};
\node at (0,0) [left] {$0$};
\draw (-0.0125,1) -- (0.0125,1) node[left] {$1$};
\draw[thick]  plot file {bayes-solution-exceeding-0.table};
\draw[dashed] (0.081,-0.0125) node[below=1em]{$0.081$} -- (0.081,1);
\draw[dashed] (0.181,-0.0125) node[below=2em]{$0.181$} -- (0.181,1);
\draw[dashed] (-0.0125,0.292) node[left=0pt]{$0.292$} -- (1,0.292);
\end{tikzpicture}
  \caption{Probability of exceeding loss threshold $L_c$ for $m=10$ and $r\sim \mathrm{Beta}(\alpha,\beta)$ with parameters corresponding to the worst case considered (bottom row of Table~\ref{tab:bayes:params}), as a function of the critical cost $L_c$.}
  \label{fig:bayes:solution}
\end{figure}

It is worth noting that these probabilities are particularly sensitive to the choices we make for the parameters of the beta distribution. However, it is straightforward to assess the influence the parameters by tabulating the probability of $L_c$ exceeding the cost of testing and termination ($181\,000$ when $m=10$) as shown in Table~\ref{tab:bayes:sensitivity}.

\begin{table}
  \centering
  \begin{tabular}{r|rrrrr}
 & $t=0.0002$ & $t=0.0004$ & $t=0.0008$ & $t=0.0016$ & $t=0.0032$ \\
\hline
 $s=200$ & $0.030$ & $0.058$ & $0.113$ & $0.211$ & $0.371$ \\
 $s=400$ & $0.036$ & $0.070$ & $0.135$ & $0.249$ & $0.428$ \\
 $s=800$ & $0.040$ & $0.079$ & $0.150$ & $0.276$ & $0.466$ \\
$s=1600$ & $0.043$ & $0.084$ & $0.160$ & $0.292$ & $0.489$ \\
$s=3200$ & $0.045$ & $0.087$ & $0.166$ & $0.301$ & $0.501$ \\
\end{tabular}

  \caption{Probability of exceeding loss threshold $L_c=181\,000$ for $m=10$ and $r\sim \mathrm{Beta}(\alpha,\beta)$ for a wide range of parameter choices. Values are shown for $L_c=182\,000$, as the probability is discontinuous at $L_c=181\,000$ (see Figure~\ref{fig:bayes:solution}).}
  \label{tab:bayes:sensitivity}
\end{table}

\subsection{Info-Gap Analysis: Maximising Robustness}

Another approach to solve our decision problem, under severe uncertainty
about the exact probability $r$ of a single animal being
infected, is to select that decision which meets a given performance
criterion, $L_c$, under the largest possible range of $r$.
Given that we have almost no information about $r$,
except that it assumes a very small but otherwise unknown value, this simple model
seems to suffice for our purpose.
Obviously, one could define many other more refined info-gap
models---and our choice of model is just one example among many.
For a much
more detailed account, see \cite{2001:benhaim}.

Specifically, for a given value of $L_c$, the
largest possible range $[0,h]$ of $r$ for which we meet our
performance criterion is characterised by
\begin{equation}
  \hat{h}(m,L_c)
  =\max_{h\ge 0}
  \left\{
    h\colon \underbrace{\max_{\substack{r\in[0,h]\\r\le 1}}L(m|r)}_{M(m,h)}\le L_c
  \right\}
\end{equation}
The value $\hat{h}(m,L_c)$, as a function of $L_c$, is called the
\emph{robustness curve}: it tells us how uncertain about $r$ we can be
for our decision $m$ still to meet a given level of performance $L_c$.

A quick Poisson approximation reveals that as long as $\exp(-nh)$ is sufficiently close to $1$ (and this holds
for sufficiently small values of $nh$) the inner maximum over
$r\in[0,h]$ is achieved at $r=h$ (also see
Figure~\ref{fig:expected:loss}: the cost increases as $r$ increases),
so
\begin{equation}
  M(m,h)=L(m|h)
\end{equation}
Obviously, $M(m,h)$ increases as the horizon of uncertainty $h$
increases, whence $\hat{h}(m,L_c)$ as a function of $L_c$ is simply
the inverse of $M(m,h)$ as a function of $h$. In other words,
plotting $M(m,h)$ as a function of $h$ for different values of $m$
effectively gives us the robustness
curves. Figure~\ref{fig:robustness:curves} depicts them.

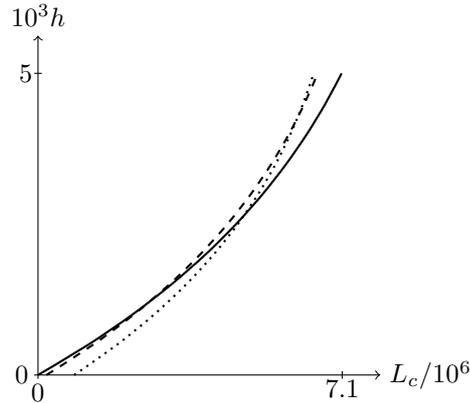
\begin{figure}
  \centering
  \begin{tikzpicture}[xscale=0.56338028169,yscale=0.8]
\draw[->] (-0.08875,0) -- (7.9875,0) node[right] {$L_c / 10^6$};
\draw[->] (0,-0.0625) -- (0,5.625) node[above] {$10^3 h$};
\node at (0,0) [below] {$0$};
\draw (7.1,-0.0625) -- (7.1,0.0625) node[below] {$7.1$};
\node at (0,0) [left] {$0$};
\draw (-0.08875,5) -- (0.08875,5) node[left] {$5$};
\draw[solid, thick]  plot[smooth] file {infectious-mu-m1-h.table};
\draw[dashed, thick]  plot[smooth] file {infectious-mu-m15-h.table};
\draw[dotted, thick]  plot[smooth] file {infectious-mu-m30-h.table};
\end{tikzpicture}
  \caption{Robustness curves $\hat{h}(m,L_c)$ as a function $L_c$ for
    test group sizes $m=1$ (solid), $m=15$ (dashed), and $m=30$ (dotted).}
  \label{fig:robustness:curves}
\end{figure}

The choices of $m$ which maximise robustness, for various values of
the critical cost $L_c$, are tabulated in
Table~\ref{table:infogap:solution}. For example, at an expected cost $L(m|r)$ of at most
$L_c=3\,000\,000$, we can safeguard against any probability of
infection $r\in[0,0.001\,479]$, by testing $10$ animals in the herd.
For comparison, the Bayesian expected cost $E(L|\alpha,\beta)$ for $t=0.0016$ is almost exactly equal to $3\,000\,000$ (last row of Table~\ref{tab:bayes:params}) when $m=10$, which is in agreement with the info-gap analysis.
Of course, it is to be noted that the actual inputs into each
decision model are different; nevertheless, both are assuming extreme
uncertainty of a similar order if we interpret the Bayesian analysis
as a worst case analysis---which we can easily do here due to the
simplicity of the model. Thus the agreement is not all that
surprising.

\begin{table}
  \centering
  \begin{tabular}{ccc}
$L_c / 10^6$ & $m^*$ & $10^3 \hat{h}(m^*,L_c)$ \\
\hline
0.5 & 2 & 0.207 \\
1.0 & 4 & 0.426 \\
1.5 & 5 & 0.661 \\
2.0 & 6 & 0.912 \\
2.5 & 8 & 1.184 \\
3.0 & 10 & 1.479 \\
3.5 & 11 & 1.803 \\
4.0 & 13 & 2.163 \\
\end{tabular}

  \caption{Info-gap choice of $m$, and corresponding horizon of uncertainty, for various values of the critical cost $L_c$.}
  \label{table:infogap:solution}
\end{table}

\subsection{Imprecise Probability Analysis: Maximality over a Partial Ordering}

There are several ways one might go about constructing an imprecise
probability model for our problem. As we have just seen, the info-gap
approach hinges on the idea of satisficing: we may start out with a
level of minimum performance that we hope to achieve, and the analysis
tells us how much uncertainty we can tolerate, at this price. One
might also interpret it conversely: for a given level of uncertainty,
the analysis tells us how much we might potentially pay, if it comes
to the worst.

Typical decision models for imprecise probabilities studied in the literature do not relate to
satisficing, yet, they do incorporate an idea similar to the info-gap
horizon of uncertainty: the imprecision of our model. Concretely,
consider the set $\mathcal{M}_h$ of all probability densities over $r$
that are zero outside $[0,h]$.\footnote{The adventurous reader may take all
  finitely additive probability measures $\mu$ on $[0,+\infty]$ with
  $\mu([0,h])=1$. We do without this complication: because all
  functions involved are continuous, those additional measures make no
  difference.}
We say that a choice $m$ \emph{dominates} a choice $m'$, and we
write $m\succ m'$ whenever the expected loss under $m$ is
strictly less than the expected loss under $m'$ over all densities
$p$ in $\mathcal{M}_h$, that is, whenever
\begin{equation}
  \int_0^\infty L(m|r)p(r) d r + \epsilon
  \le \int_0^\infty L(m'|r)p(r) d r
\end{equation}
for all probability densities $p$ in $\mathcal{M}_h$ and some
$\epsilon>0$.

Because $L(m|r)$ is continuous in $r$ for every $m$, this happens if and only if
\begin{equation}
  \min_{r\in[0,h]}\left[L(m'|r)-L(m|r)\right]>0
\end{equation}
Note that the $\min_{r\in[0,h]}$ operator can be thought of as a lower
expectation operator, or \emph{lower prevision} $\lpr_h$---we will come
back to this in Section~\ref{sec:discussion}.

One can easily prove that $\succ$ is a partial order, whence, a
sensible way to choose $m$ is to pick one which is not dominated by any other option, or in other words, which is \emph{maximal}. The idea of
choosing undominated options goes back at least to Condorcet
\cite[pp.~lvj--lxix, 4.$^{e}$ Exemple]{1785:condorcet}; also see
\cite[p.~55, Eq.~(1)]{1977:sen},
\cite[Sections~3.7--3.9]{1991:walley}, and
\cite{2007:troffaes:decision:intro} for further discussion.
Maximality has also been used in robust Bayesian models
\cite[\S 10.4]{2000:rios:bayesian:sens:anal},
under slightly different terminology.

Given our
partial order, one can easily show that an option $m$ is maximal if
and only if
\begin{equation}\label{eq:max:calculation}
  \min_{m'\in\{0,1,\dots,n\}}\max_{r\in[0,h]}\left[L(m'|r)-L(m|r)\right]\ge 0
\end{equation}
The inner maximum is almost always achieved at either $r=0$ or $r=h$,
simplifying practical calculations substantially.
Table~\ref{table:maximal:solution} depicts these values for all
choices of $m$, and varying values of $h$.
For ease of comparison with the info-gap solution,
we have chosen the same values of $h$ as those listed in
Table~\ref{table:infogap:solution}.

\begin{table}
  \centering
  \begin{tabular}{r|rrrrrrrr}
 & \multicolumn{8}{c}{$10^3 h$} \\
$m$ & $0.207$ & $0.426$ & $0.661$ & $0.912$ & $1.184$ & $1.479$ & $1.803$ & $2.163$ \\
\hline
$0$ & $-0.9$ & $-0.9$ & $-0.9$ & $-0.9$ & $-0.9$ & $-0.9$ & $-0.9$ & $-0.9$ \\
$1$ & $1.1$ & $1.1$ & $1.1$ & $1.1$ & $1.1$ & $1.1$ & $1.1$ & $1.1$ \\
$2$ & $1.4$ & $3.1$ & $3.1$ & $3.1$ & $3.1$ & $3.1$ & $3.1$ & $3.1$ \\
$3$ & $-0.6$ & $2.1$ & $4.9$ & $5.1$ & $5.1$ & $5.1$ & $5.1$ & $5.1$ \\
$4$ & $-3.1$ & $0.1$ & $2.9$ & $5.9$ & $7.1$ & $7.1$ & $7.1$ & $7.1$ \\
$5$ & $-7.7$ & $-1.9$ & $0.9$ & $3.9$ & $7.0$ & $9.1$ & $9.1$ & $9.1$ \\
$6$ & $-14.3$ & $-5.8$ & $-1.1$ & $1.8$ & $5.0$ & $8.4$ & $11.1$ & $11.1$ \\
$7$ & $-22.9$ & $-11.8$ & $-4.3$ & $-0.2$ & $2.9$ & $6.3$ & $9.9$ & $13.1$ \\
$8$ & $-33.4$ & $-19.7$ & $-9.5$ & $-2.4$ & $0.9$ & $4.2$ & $7.9$ & $11.8$ \\
$9$ & $-46.0$ & $-29.7$ & $-16.6$ & $-6.7$ & $-1.1$ & $2.2$ & $5.8$ & $9.7$ \\
$10$ & $-60.6$ & $-41.7$ & $-25.9$ & $-13.0$ & $-4.3$ & $0.1$ & $3.7$ & $7.6$ \\
$11$ & $-77.2$ & $-55.6$ & $-37.1$ & $-21.3$ & $-9.5$ & $-1.9$ & $1.7$ & $5.6$ \\
$12$ & $-95.8$ & $-71.6$ & $-50.3$ & $-31.6$ & $-16.8$ & $-5.9$ & $-0.4$ & $3.5$ \\
$13$ & $-116.4$ & $-89.6$ & $-65.6$ & $-44.0$ & $-26.1$ & $-11.9$ & $-2.9$ & $1.4$ \\
$14$ & $-139.1$ & $-109.7$ & $-82.9$ & $-58.4$ & $-37.4$ & $-20.0$ & $-7.4$ & $-0.7$ \\
$15$ & $-163.7$ & $-131.7$ & $-102.2$ & $-74.8$ & $-50.8$ & $-30.1$ & $-14.1$ & $-3.5$ \\
\end{tabular}

  \caption{Result of Eq.~\eqref{eq:max:calculation} (divided by a factor $10^3$ for everything to fit in the table). A positive value means that the corresponding choice of $m$ is optimal for the given horizon of uncertainty $h$.}
  \label{table:maximal:solution}
\end{table}

\section{Discussion}
\label{sec:discussion}

It is already well known that robust Bayesian models and imprecise probability
models are, for the most part, mathematically equivalent \cite[\S 5.9,
pp.~253--258]{1991:walley}.
Therefore, in the following, we will focus on info-gap and
maximality---the latter being used for both imprecise
and robust Bayesian models.

Interestingly, in our example, info-gap and maximality give essentially the same
result, with maximality refining the picture slightly:
for a given horizon of uncertainty $h$,
the maximal solutions are $\{1,\dots,m^*\}$,
where $m^*$ is the info-gap solution.
The most notable result is that all info-gap solutions are maximal.
Is this a coincidence? Formulating info-gap theory in terms of
lower previsions, we show that this holds
for arbitrary info-gap models and arbitrary lower previsions,
subject to the mild and usually satisfied conditions of
Theorem~\ref{thm:infogap:iff:gammaminimax} (for $\Gamma$-maximin) and
Theorem~\ref{thm:infogap:implies:maximal} (for maximality).

\subsection{Info-Gaps for Imprecise Probabilities}

Let $\omega\in\Omega$ be an uncertain parameter of interest---$\Omega$
can be an arbitrary set.
We must select a decision $d$ from a finite set $D$.
The loss function $L(d,\omega)$ represents
the loss (in utiles) if we choose $d$ and $\omega$ obtains.

Info-gap theory starts out with a family of nested sets
$U_h$ of $\Omega$, where $h$ is a non-negative parameter called
the \emph{horizon of uncertainty} and $U_h\subseteq U_{h'}$ whenever
$h\le h'$. In our example, $U_h$ was simply $[0,h]$.
Following that example, we saw that a very natural way to model these
nested sets $U_h$ in terms of sets of probabilities goes by way of a
\emph{vacuous model} $\mathcal{M}_h$, that is, the set of all
probability densities that are zero outside $U_h$.

If we denote the upper expectation induced by $\mathcal{M}_h$ by
$\upr_h$
(i.e. the pointwise lowest upper bound for the set of expectation operators
associated to $\mathcal{M}_h$),
then, formally, we define the info-gap solution
$D^*(L_c)\subseteq D$ at satisficing level $L_c$ as:
\begin{align}
  \hat{h}(d,L_c)&\coloneqq\max\left\{h\colon\upr_h(L(d,\cdot))\le L_c\right\} \\
  D^*(L_c)&\coloneqq\arg\max_{d\in D}\hat{h}(d,L_c)
\end{align}
Note that $D^*(L_c)$ will
usually be a singleton (or, the empty set).

Also note that the first equation may not have a solution:
this happens when $\upr_0(L(d,\cdot))>L_c$, that is, when
$d$ is infeasible even if we are as certain as can be ($h=0$).

Now, from the point of view of imprecise
probability, there is no compelling reason to restrict ourselves to
vacuous models. In fact, we can allow $\mathcal{M}_h$ to be any set of
probability densities on $\Omega$---we already picked a more general set in our Bayesian analysis in Section~\ref{sec:bayesian}---under one restriction: a close
inspection of the theory reveals that a crucial property that the
info-gap model relies on is that the worst case cost,
$\upr_h(L(d,\cdot))$ is non-decreasing as the horizon of uncertainty $h$
increases. Whence, we logically impose that
$\mathcal{M}_h\subseteq\mathcal{M}_{h'}$ whenever $h<h'$.

So, instead of starting out from a family of nested
subsets $U_h$ of $\Omega$, we start out from a family of nested sets
$\mathcal{M}_h$ of probability densities on $\Omega$.
For instance, in our example, the uncertainty was over the values of $r$, so
$\mathcal{M}_h$ would be some arbitrary set of probability
distributions for $r$. In the Bayesian analysis, we restricted
$\mathcal{M}_h$ to the Beta family, for computational convenience.

One can of course
interpret this again as an info-gap model,
where the uncertain parameter is now the probability density over
$\Omega$---also
see \cite[pp.~1062--1063]{2009:benhaim:cholesterol} for an informal
discussion of this approach.
The imprecise
Dirichlet model \cite{1996:walley::idm} is an example of such
family (with $h=1/s$).
For another example, see \cite{2008:ferson::pboxes:infogap} for a discussion of
nested sets of p-boxes and the resulting info-gap analysis.

\subsection{Main Result}

The next result links the info-gap solution to the so-called
$\Gamma$-minimax\footnote{$\Gamma$-minimax minimises the upper
  expectation of the loss \cite{1945:wald,1989:gilboa::maximin}.} solution.
See \cite[p.~1061,
Fig.~14]{2009:benhaim:cholesterol} for an informal discussion of a
very similar equivalence between info-gap and minimax.
Interestingly, \cite[p.~4, Table~1]{2006:sniedovich:infogapproof}
constructs a maximin model that is fully equivalent
to an info-gap model. The result below is quite different,
as we change neither variables nor loss function.

Effectively, we show that, if certain fairly mild conditions are satisfied,
the info-gap solution coincides with the $\Gamma$-minimax
solution---parametrized over the horizon of uncertainty $h$---of
exactly the same problem.
In fact, such approach has already been used as a technique
to solve info-gap problems
(see for instance \cite[pp.~1690--1691]{2006:moilanen:infogap}).
Below, we identify sufficient conditions 
for this to work, and provide two counterexamples
in cases where one of these conditions fails.

For convenience, we denote the $\Gamma$-minimax loss at horizon $h\in\reals^+$
by $L^*(h)$:\footnote{We set $\reals^+\coloneqq\{x\in\reals\colon x\ge 0\}$.}
\begin{equation}\label{eq:def:lstar}
  L^*(h)\coloneqq\min_{d\in D}\upr_{h}(L(d,\cdot)).
\end{equation}
Note that $L^*$ is non-decreasing as a function of $h$,
because each $\upr_{h}(L(d,\cdot))$ is.
We will be interested in the \emph{right derivative} of $L^*$ at a point $h$:
\begin{equation}
  \partial_+ L^*(h)\coloneqq \lim_{\substack{h'\to h\\ h'>h}}\frac{L^*(h')-L^*(h)}{h'-h}
\end{equation}
Because $L^*$ is non-decreasing, $\partial_+ L^*(h)$ is non-negative
and well defined (possibly $+\infty$) for every $h\in\reals^+$.

\begin{theorem}\label{thm:infogap:iff:gammaminimax}
  Let $h\in\reals^+$ and $L_c\in\reals$.
  The info-gap solution $D^*(L_c)$ coincides with
  $\Gamma$-minimax solution with respect to
  $\upr_{h}$, that is,
  \begin{equation}
    D^*(L_c)=\arg\min_{d\in D}\upr_{h}(L(d,\cdot)),
  \end{equation}
  whenever the following conditions are satisfied:
  \begin{align}
    \label{eq:lc:increasing}
    \partial_+L^*(h)&>0, \text{ and} \\
    \label{eq:lc:and:h}
    L^*(h)&=L_c
  \end{align}
\end{theorem}
\begin{proof}
  By definition, $d^*\in D^*(L_c)$ whenever, for all $d\in D$,
  \begin{equation}
    \hat{h}(d^*,L_c)\ge \hat{h}(d,L_c)
  \end{equation}
  By definition of $\hat{h}(d,L_c)$, this is equivalent to saying that
  \begin{equation}
    \left\{h'\colon\upr_{h'}(L(d^*,\cdot))\le L_c\right\}
    \supseteq
    \cup_{d\in D}
    \left\{h'\colon\upr_{h'}(L(d,\cdot))\le L_c\right\}
  \end{equation}
  Rewriting the above expression, we have, equivalently,
  \begin{equation}
    \left\{h'\colon\upr_{h'}(L(d^*,\cdot))\le L_c\right\}
    \supseteq
    \left\{h'\colon\min_{d\in D}\upr_{h'}(L(d,\cdot))\le L_c\right\}
  \end{equation}
  or, by Eq.~\eqref{eq:def:lstar},
  \begin{equation}
    \left\{h'\colon\upr_{h'}(L(d^*,\cdot))\le L_c\right\}
    \supseteq
    \left\{h'\colon L^*(h')\le L_c\right\}
  \end{equation}
  By Eq.~\eqref{eq:lc:increasing} and the non-decreasingness of $L^*$,
  it is easily seen that $L^*(h')>L^*(h)$ for all $h'>h$.
  Moreover, by Eq.~\eqref{eq:lc:and:h}, $L_c=L^*(h)$.
  Concluding, the set on the right hand side is a fancy
  way of writing $[0,h]$. Therefore, the above is equivalent to
  \begin{equation}
    \upr_{h}(L(d^*,\cdot))\le L_c
  \end{equation}
  Once more by Eqs.~\eqref{eq:lc:and:h} and~\eqref{eq:def:lstar}, this is equivalent to saying
  that $d^*$ is a $\Gamma$-minimax solution with respect to
  $\upr_{h}$.
\end{proof}

Interestingly, for given $L_c$ such that
\begin{equation}
  \min_{d\in D}\upr_0(L(d,\cdot))\le L_c\le \min_{d\in D}\upr_\infty(L(d,\cdot))
\end{equation}
it holds that Eq.~\eqref{eq:lc:and:h} has a unique
solution for $h\ge 0$ whenever $L^*$ is strictly increasing \emph{and continuous} in
$h$.
This means that we are effectively free to choose $L_c$ under the
additional assumption of continuity.

To see why we are not free to
choose $L_c$ when continuity is not
satisfied---and violate Eq.~\eqref{eq:lc:and:h}---imagine for instance that:
\begin{align}
  \upr_{h}(L(d_1,\cdot))&=
  \begin{cases}
    h & \text{if }h\le 1 \\ 
    3+h & \text{if }h> 1
  \end{cases}
  &
  \upr_{h}(L(d_2,\cdot))&=
  \begin{cases}
    1+h & \text{if }h\le 1 \\
    4+h & \text{if }h> 1
  \end{cases}
\end{align}
Then, for $L_c=3$, we have that $D^*(3)=\{d_1,d_2\}$ because $\hat{h}(d,3)=1$ for both
$d_1$ and $d_2$, yet obviously $d_1$ is $\Gamma$-minimax (it could even
be uniformly dominated by $d_2$). Effectively, this is simply a technical
limitation of the info-gap model, as any reasonable person would 
probably agree with the $\Gamma$-minimax solution.

What happens if Eq.~\eqref{eq:lc:increasing} is violated?
Imagine for instance that:
\begin{align}
  \upr_{h}(L(d_1,\cdot))&=
  \begin{cases}
    0 & \text{if }h\le 1 \\ 
    h-1 & \text{if }h> 1
  \end{cases}
  &
  \upr_{h}(L(d_2,\cdot))&=
  \begin{cases}
    0 & \text{if }h\le 2 \\
    h-2 & \text{if }h> 2
  \end{cases}
\end{align}
With this choice, $L^*$ is continuous, but not strictly increasing.
For $h=1$, we have that $\arg\min_{d\in D}\upr_{h}(L(d,\cdot))=\{d_1,d_2\}$
because $\upr_{h}(L(d_1,\cdot))=\upr_{h}(L(d_2,\cdot))=0$ for $h=1$.
Yet, for $L_c=L^*(1)=0$ we have that $\hat{h}(d_1,L_c)=1<\hat{h}(d_2,L_c)=2$,
so only $d_2$ is optimal according to the info-gap criterion.
This example uncovers a technical
limitation of the $\Gamma$-maximin model, as for this case,
any reasonable person would 
probably agree with the more robust info-gap solution.

Now, it is well known that every $\Gamma$-minimax solution is also
maximal (see for instance \cite{2007:troffaes:decision:intro}),
whence, we conclude:

\begin{theorem}\label{thm:infogap:implies:maximal}
  Suppose $\partial_+L^*(h')>0$ for all $h'\in [0,h]$.
  Then, for all $h'\in[0,h]$, every info-gap decision $d^*\in D^*(L^*(h'))$ is
  maximal with respect to $\upr_{h}$:
  \begin{equation}
    \bigcup_{h'\in[0,h]}D^*(L^*(h'))
    \subseteq
    \{d\in D\colon (\forall d'\in D)(\upr_h(L(d',\cdot)-L(d,\cdot))\ge 0)\}
  \end{equation}
\end{theorem}
\begin{proof}
  Use the preceding theorem, and note that every $\Gamma$-minimax
  with respect to $\lpr_{h'}$ is maximal with respect to $\lpr_h$,
  provided that $h'\in[0,h]$.
\end{proof}

Again, if in addition $L^*$ is continuous on $[0,h]$, then the range
for $L^*(h')$ in the above theorem is simply an interval:
\begin{equation}
  \{L^*(h')\colon h'\in[0,h]\}
  =
  \left[\min_{d\in D}\upr_0(L(d,\cdot)),\min_{d\in D}\upr_{h}(L(d,\cdot))\right].
\end{equation}

Summarising, Theorem~\ref{thm:infogap:iff:gammaminimax} provides sufficient
conditions for the info-gap solution, for fixed values of
$L_c$ and $h$, to be equal to the $\Gamma$-minimax
solution:
proponents of either approach are `observationally equivalent' \cite[Sec.~7]{2009:benhaim:cholesterol}.

Theorem~\ref{thm:infogap:implies:maximal} shows that a full fledged
info-gap analysis, varying the horizon of uncertainty along an
interval, yields an elegant approach to capture maximal
solutions. In our example, we actually find \emph{all} maximal
options---in general this may not be the case. Still, it shows
that an info-gap analysis can be of value even if maximality
is the final goal:
\begin{itemize}
\item an info-gap analysis might give a rough idea of the size of the
  maximal set (in particular, it provides a lower bound for it),
\item the analysis can be an appealing way to represent the maximal
  solution graphically (as in Figure~\ref{fig:robustness:curves}), and
\item as robustness curves show the trade-off between uncertainty and
  cost, they are also obviously useful in the process of elicitation.
\end{itemize}

\section{Conclusion}
\label{sec:conclusion}

We constructed a simple model for inspecting animal herds for dangerous
exotic infections, building further on the work of Moffitt et
al. \cite{2008:moffitt}. We solved the problem using three popular
decision methodologies suited for dealing with severe uncertainty:
Bayesian analysis, info-gap analysis, and imprecise probability theory (maximality and $\Gamma$-minimax). We found that, in this example, the solutions of the info-gap and imprecise
models essentially coincide, although the way they arrive at it is very different.

We explored the theoretical link between info-gap theory, $\Gamma$-minimax,
and maximality. We established that, under rather general conditions,
every info-gap solution is maximal. Therefore, the set of maximal
options can be inferred at least partly, and sometimes wholly, from an info-gap analysis.
Consequently, robustness curves also make sense in an imprecise
probability (or, robust Bayesian) context, for exploring maximal options, and for
elicitation, when studying the trade-off between uncertainty and cost
that is often of interest to decision makers.

\section*{Acknowledgements}

The authors thank Kirsty Hinchliff and Ben Powell, who have been
involved with an embryonic draft of this paper. The paper has also
benefited greatly from discussions with Yakov Ben-Haim, Frank
Coolen, and Moshe Sniedovich, to who we extend our sincerest thanks.

\bibliographystyle{amsplainurl}
\bibliography{all}

\end{document}